\documentclass{svproc}
\usepackage{amsmath}
\usepackage{amssymb}

\newcommand{\ket}[1]{\lvert #1\rangle}

\usepackage{subfig}
\usepackage{subfloat}

\usepackage{graphicx}
\usepackage{type1cm}
\usepackage{dcolumn}
\usepackage{bm}
\usepackage{hyperref}
\usepackage{tikz}
\usepackage{dsfont}
\usetikzlibrary{calc}

\usepackage{pifont}
\newcommand{\cmark}{\ding{51}}%
\newcommand{\xmark}{\ding{55}}%

\usepackage{orcidlink}

\usepackage{qcircuit}

\usepackage{amsmath}
\usepackage{upgreek}
\usepackage{physics}
\usepackage{isomath}
\usepackage{placeins}

\usepackage{listings}
\newcolumntype{M}[1]{>{\centering\arraybackslash}m{#1}}

\begin{document}


\title{Increasing Interference Detection in Quantum Cryptography using the Quantum Fourier Transform}
\titlerunning{Quantum Cryptography using the Quantum Fourier Transform}

\author{Nicholas J.C. Papadopoulos\orcidlink{0000-0002-6357-0030}\inst{1} \and Kirby Linvill\orcidlink{0000-0001-8835-3377}\inst{1}}

\institute{Department of Computer Science, University of Colorado Boulder, Boulder CO 80309, USA\\
\email{\{Nicholas.Papadopoulos, Kirby.Linvill\}@colorado.edu}}


\maketitle

\date{\today}
\begin{abstract}
Quantum key distribution (QKD) and quantum message encryption protocols promise a secure way to distribute information while detecting eavesdropping.
However, current protocols may suffer from significantly reduced eavesdropping protection when only a subset of qubits are observed by an attacker.
In this paper, we present two quantum cryptographic protocols leveraging the quantum Fourier transform (QFT) and show their higher effectiveness even when an attacker measures only a subset of the transmitted qubits.
The foremost of these protocols is a novel QKD method that leverages this effectiveness of the QFT while being more practical than previously proposed QFT-based protocols, most notably by not relying on quantum memory.
We additionally show how existing quantum encryption methods can be augmented with a QFT-based approach to improve eavesdropping detection.
Finally, we provide equations to analyze different QFT-based detection schemes within these protocols so that protocol designers can make custom schemes for their purpose.
\end{abstract}

\keywords{Quantum cryptography, key distribution, quantum Fourier transform.}

\section{Introduction}
Quantum key distribution (QKD) is a paradigm that employs quantum principles to exchange a secret key which can then be used to establish secure communication channels.
Unlike classical key agreement algorithms, QKD schemes allow participants to detect when an eavesdropper listens in on the exchanged key information.
Importantly, this eavesdropping detection relies on well-established properties of quantum mechanics, namely the no-cloning theorem~\cite{no_cloning} and state collapse upon measurement, instead of on the assumed hardness of problems such as the discrete log problem used in the classical Diffie-Hellman key exchange algorithm~\cite{diffie_hellman} or the semiprime factoring problem used in the Rivest–Shamir–Adleman (RSA)~\cite{rsa} public encryption algorithm.
This reliance on quantum mechanical properties is especially advantageous given that Shor~\cite{Shor_1994} devised a way to efficiently break these classical cryptographic schemes using a sufficiently large quantum computer.

The BB84 protocol, introduced by Bennett and Brassard in 1984~\cite{bb84}, was the first such QKD protocol.
Several other QKD schemes have been developed that rely on the same underlying principles.
Unfortunately, in many QKD schemes, including BB84, subsequently transmitted qubits are independent of previous qubits. 
This independence leaves these schemes susceptible to partial information leakage and the concept of a ``many-copies'' attack introduced in Sec.~\ref{sec:many-copies}, which is solely an issue in quantum-based protocols.
Partial information leakage allows an eavesdropper to significantly reduce the search space without needing to extract the whole message, while the ``many-copies attack'' allows an adversary to exploit vulnerabilities in conventional QKD methods by flooding the system with replicated quantum states.
Through the ``many-copies attack'' attack, an eavesdropper could potentially extract the entire message with certainty and without being detected.

Our work introduces a novel two-pass QKD approach employing the quantum Fourier transform (QFT)~\cite{nielsen_chuang_2010} to better resist these attacks.
Additionally, we extend the application of the QFT to direct message encryption, presenting a three-pass protocol that combines the benefits of confidential communication with an increase in eavesdropping detection.
As part of our contributions, we explore various eavesdropping detection schemes within QFT-based cryptographic protocols and discuss the implications of different encoding strategies, considering scenarios where carefully chosen qubits are measured by an eavesdropper.
This analysis offers insights into the design of robust eavesdropping detection schemes tailored to specific application requirements.

\paragraph*{Contributions} Our contributions are summarized as follows:
\begin{itemize}
    \item We introduce a novel QKD protocol that uses the QFT to improve eavesdropping detection (Sec.~\ref{sec:qft-qkd}).
    This scheme also provides an entirely quantum QKD protocol, allowing for eavesdropping detection without the need for an authenticated classical channel that other schemes require to protect against intercept-resend attacks.
    \item We show how Kak's well-known quantum message encryption scheme~\cite{threeStage} can be augmented  with the QFT, thereby improving its eavesdropping detection (Sec.~\ref{sec:qft-enc}).
    \item We introduce the many-copies attack, in which an attacker is able to use partial observations to reconstruct a whole shared key (Sec.~\ref{sec:many-copies}).
    \item We provide equations to analyze the probability of detecting an eavesdropper when using QFT-based eavesdropping detection schemes, which we call verification schemes (Sec.~\ref{sec:probability}).
    \item We compare different key examples of QFT-based verification schemes to show that our protocol is more effective in detecting eavesdroppers than other protocols, and we show how one might want to cater such a scheme to their specific case. (Sec.~\ref{sec:schemes}).
\end{itemize}

\section{Related Work}

QKD has been a booming field since the foundational BB84~\cite{bb84} protocol was developed.
Ekert subsequently proposed a protocol based on entanglement rather than the impossibility of perfectly discriminating non-orthogonal states~\cite{e91}.
However, this entanglement is only between qubits that represent the same bit of information.
Ekert's protocol lacks entanglement between subsequent qubits representing different bits of information.
As a result, Ekert's protocol is just as susceptible to the partial information leakage and ``many-copies'' attack as BB84.
Other, more recent QKD approaches derived from these works are reviewed in more depth by~\cite{qcSurvey}.
The theoretical development of these protocols has been accompanied by implementations on real hardware~\cite{Lucamarini:13,Zhang_2019,Korzh_2015}.

Unfortunately, real implementations have inspired attacks against the necessarily imperfect implementations of these schemes~\cite{qkd_security}.
For example, the photon number splitting (PNS) attack~\cite{pns_attack} takes advantage of the fact that physical implementations of QKD devices can emit more than one photon per pulse. Mitigations against these PNS attacks include the use of decoy states to verify that photon loss is within expected parameters~\cite{hwang2003decoy,lo2005decoy,wang2005decoy}. 
Since these mitigations against physical imperfections and side-channels are typically orthogonal to the underlying scheme, we do not focus on such mitigations for the scheme we introduce.

A few QKD schemes that entangle subsequent qubits have been proposed.
Similar to the QKD scheme we propose in this paper, Nagy, Akl, and Kershaw introduced a QKD method utilizing the quantum Fourier transform (QFT) that improved upon BB84-like protocols by causing entanglement if there is interference or a 3rd-party observation~\cite{nagy2008key}.
However, the proposed method of Nagy et al. relies on quantum memory, which may be seen as infeasible or unreliable with present technology.
In this paper, we propose a method that also uses the QFT, and therefore leverages its benefits such as a high sensitivity to 3rd-party observations, without relying on quantum memory to store qubits for some indeterminate amount of time.

Other methods using the QFT for QKD have also been proposed.
Xing-Yu Yan et al.~\cite{quditQft} proposed a QFT-based scheme to encode high-dimensional data using qudits.
They achieve higher information-density rates than other QKD schemes but rely on a pre-shared $d$-dimensional quantum state, which is not a standard assumption in QKD schemes as this would require trusted or authenticated quantum channels.
Xiaoqing Tan et al.~\cite{qkdQftWorkshop} proposed another scheme that distributes EPR pairs to two parties and applies the QFT to the qubits of one of the parties.
However, this protocol requires quantum memory, similar to Nagy et al., both while Bob waits for Alice to communicate the positions of the decoy photons and while Alice waits for Bob to communicate his measurement bases.

In addition to providing a QKD scheme based on the QFT, we provide a scheme in this paper to directly send encrypted messages.
Kak~\cite{threeStage} proposed a three-stage protocol that allows Alice to send an encrypted message directly to Bob. 
This eliminates the need for classical messaging while simultaneously protecting the confidentiality of the sent data and allowing for the detection of Eve's presence.
We propose a modified instance of Kak's protocol using the QFT, allowing for a higher sensitivity of detection under certain attacks while still directly sending encrypted messages.

\section{Background: Notable Attacks on QKD Protocols}

\subsection{Partial Information Leakage} \label{sec:partial-information}

QKD protocols aim to detect eavesdroppers on a quantum channel and therefore must be analyzed with respect to a quantum adversary, typically called Eve.
One of the simplest attacks Eve can carry out is the intercept-resend attack, in which Eve measures photons in transit in a chosen basis and then passes on any prepared state of her choosing, often the measured state. 
The BB84 protocol protects against this attack by having its participants, typically called Alice and Bob, choose a random subset of exchanged qubits to verify that they indeed match the expected measurement values.
This verification step is conducted over an authenticated classical channel after the qubits have all been exchanged.
In the case of BB84, an attacker is expected to be detected with probablity $1 - (.75)^{|V|}$ where $V$ is the set of verification qubits, meaning qubits used solely for detecting Eve's presence, and $|V|$ is the number of verification qubits.
However, this protection assumes the quantum adversary measures all the transmitted qubits.

An attacker could alternatively choose to only measure a subset of the qubits transmitted over the wire, resulting in only partial information leakage in exchange for a reduced chance of detection.
For the BB84 protocol, the expected chance of detection is $1 - (.75)^{Pr[q \in V] * |V|}$ where $Pr[q \in V]$ is the probability that a qubit $q$ Eve measures is a verification qubit.
In the average case, this results in a decreased chance of detection in exchange for partial information.
However, in the worst case Eve will not measure any verification qubits resulting in a probability of detection of $0$.
Furthermore, if the qubits that will be used as verification bits are predictable, perhaps by relying on a weak pseudo-random number generator~\cite{cwe-insufficient-randomness,cwe-prng} or perhaps by relying on a deterministic verification scheme, Eve can ensure she always only measures non-verification qubits and therefore always avoid detection.

Figure~\ref{fig:bb84-partial-leakage} shows an example of such an attack on the BB84 protocol.
Alice first generates a random message and encodes each qubit in a random choice of basis.
She then sends the qubits to Bob.
Eve chooses to measure half of the qubits.
Alice and Bob then announce the bases they transmitted/measured in and discard the bits for which they used different bases.
Alice and Bob then publicly choose the verification bits out of the remaining bits.
In this case Eve was lucky, and none of the bits she measured were verification bits, so she is not detected.

Alice and Bob can optionally go through a privacy amplification process~\cite{privacy-amplification} to reduce the information Eve learns about the exchanged key.
However, if Eve measured all of the exchanged key bits, she would on average learn half of the bits that she measured.
This is because she has a 50\% of randomly choosing the same basis in which Alice transmitted a qubit.
Since she knows about half of the exchanged key bits, Eve will then be able to significantly cut down the possible keys after privacy amplification from $2^{|k|}$ to $2^{|k|/2}$, where $|k|$ is the length of the derived key $k$ in bits, all while avoiding detection.

\begin{figure}
    \centering
    \includegraphics[width=\linewidth]{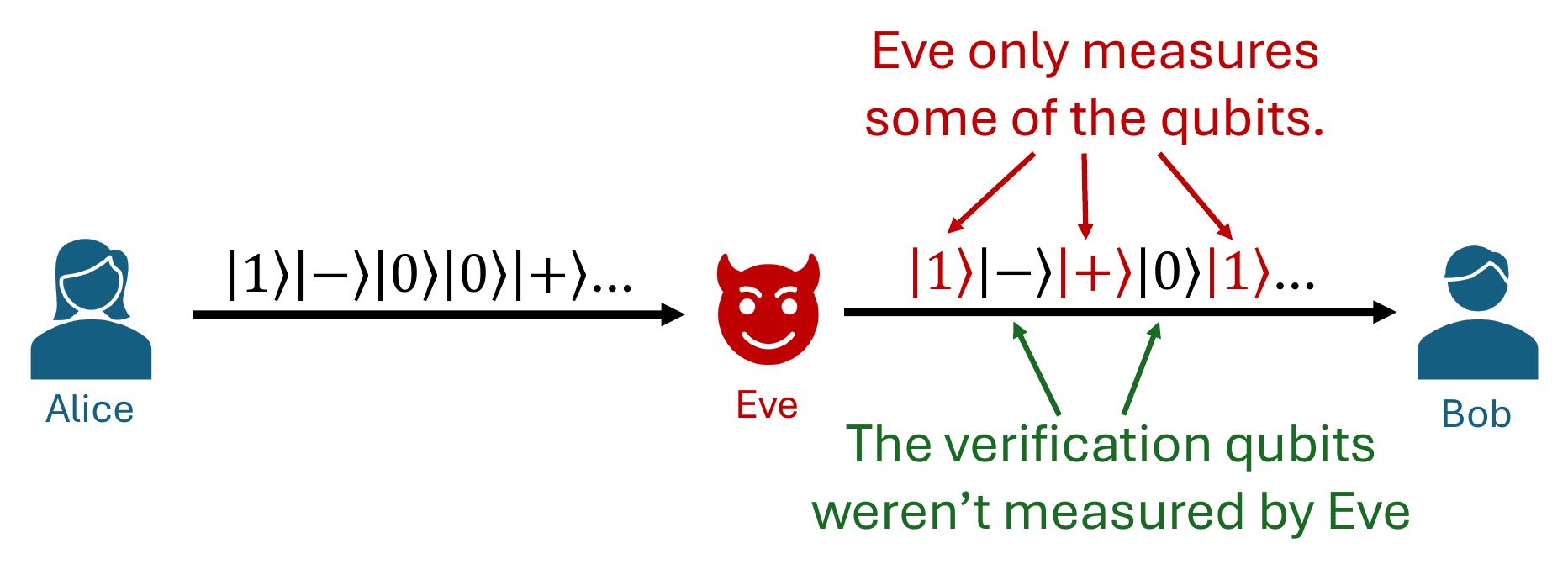}
    \caption{Example partial leakage from the BB84 protocol. Eve only measures some of the qubits Alice sends to Bob. She measures none of the verification qubits and therefore has no chance of detection.}
    \label{fig:bb84-partial-leakage}
\end{figure}

\subsection{The Many-Copies Attack}\label{sec:many-copies}

We introduce the concept of an attack where many copies of some state may be sent by Alice or Bob. 
Despite the no-cloning theorem~\cite{no_cloning}, which states it is impossible to clone an unknown quantum state, this scenario may arise by exploiting deterministic encodings of data into a known quantum state.
This could be done by any number of means, including maliciously by Eve planting a bug in their computers to create and send many copies of a transmission, or it could be due to an honest user-error or misconfiguration in either hardware or software. 
We call this the many-copies attack, and it can allow Eve to obtain information about bits in BB84 and similar protocols without detection.
By measuring some non-verification qubit in many different bases over many copies of the quantum state, she would be able to extract the intended value with certainty without detection.
Eve could do this for, in the worst case, each non-verification qubit sequentially across many copies of the state to retrieve the whole message undetected.
This is similarly a problem in Kak's three-stage protocol~\cite{threeStage}, where Eve could potentially retrieve $U_A^\dagger$ using similar techniques and apply it to the initial state to retrieve the original message.

\begin{table}[h]
    \caption{An example of the many-copies attack on the BB84 protocol. Eve remains undetected while finding the value of the secret bit.}
    \begin{center}
    \begin{tabular}{|c||c|c||c|c||}
        \hline
         & \multicolumn{2}{c||}{Original Copy} & \multicolumn{2}{c||}{Copy 2} \\
        \hline
         & Secret & Verification & Secret & Verification \\
        \hline
         \textbf{Alice's bits} & 1 & 0 & 1 & 0 \\
         \hline
         \textbf{Alice's sending basis} & + & $\times$ & + & $\times$ \\
         \hline
         \hline
         \textbf{Eve's measuring basis} & + & & + & \\
         \hline
         \textbf{Eve's measurement} & 1 & & 1 & \\
         \hline
         \hline
         \textbf{Bob's measuring basis} & + & $\times$ & $\times$ & +  \\
         \hline
         \textbf{Bob's measurement} & 1 & 0 & 1 & 1 \\
         \hline
         \textbf{Measurement bases match} & \cmark & \cmark & \xmark & \xmark \\
         \hline
         \textbf{Verification passes} &  & \cmark &  &  \\
         \hline
    \end{tabular}
    \begin{tabular}{|c||c|c||c|c||}
        \hline
         & \multicolumn{2}{c||}{Copy 3} & \multicolumn{2}{c||}{Copy 4} \\
        \hline
         & Secret & Verification & Secret & Verification \\
        \hline
         \textbf{Alice's bits} & 1 & 0 & 1 & 0 \\
         \hline
         \textbf{Alice's sending basis} & + & $\times$ & + & $\times$ \\
         \hline
         \hline
         \textbf{Eve's measuring basis} & $\times$ & & $\times$ & \\
         \hline
         \textbf{Eve's measurement} & 1 & & 0 & \\
         \hline
         \hline
         \textbf{Bob's measuring basis} & + & + & $\times$ & $\times$ \\
         \hline
         \textbf{Bob's measurement} & 0 & 1 & 0 & 0 \\
         \hline
         \textbf{Measurement bases match} & \cmark & \xmark & \xmark & \cmark \\
         \hline
         \textbf{Verification passes} & & & & \cmark \\
         \hline
    \end{tabular}
    \end{center}
    \label{tab:manyCopies}
\end{table}
Table~\ref{tab:manyCopies} shows an example of this attack.
Alice and Bob try to generate a 1-bit secret key and use a single verification bit to try to detect Eve.
Four total copies of the message from Alice to Bob are sent, and Eve measures only the secret key bit.
One can see that this has no effect on the verification bit, and so Eve will remain undetected.
Eve, on the other hand, measures each basis twice.
She sees that one basis is consistent while the other has multiple outputs.
She can then conclude that Alice measured in the basis that was consistent, and the intended bit value was the value that was measured from that basis.

\section{Memory-free Quantum Cryptographic Protocols using the QFT}

We now present a two-pass quantum key distribution protocol using the quantum Fourier transform for the purpose of quantum key distribution.
The protocol allows Alice to share a secret key of $|k|$ random bits with Bob such that they have a higher probability of detecting an eavesdropper Eve compared to BB84, regardless of which qubits Eve measures.
Unlike the scheme presented in \cite{nagy2008key}, our method does not require quantum memory, i.e., holding the quantum state until a later time, such as when further classical communication is received.
Instead, all state modifications and measurements in this protocol are done upon reception of the state.
This scheme is therefore much more feasible to realize using current and near-future technology than schemes that require qubit storage.

\subsection{Quantum Key Distribution Protocol} \label{sec:qft-qkd}
The protocol consists of four steps: phase scrambling, key encoding, key decoding, and key extraction. The first three of these steps can be carried out using the circuit shown in Fig.~\ref{fig:twoPassCircuit}.
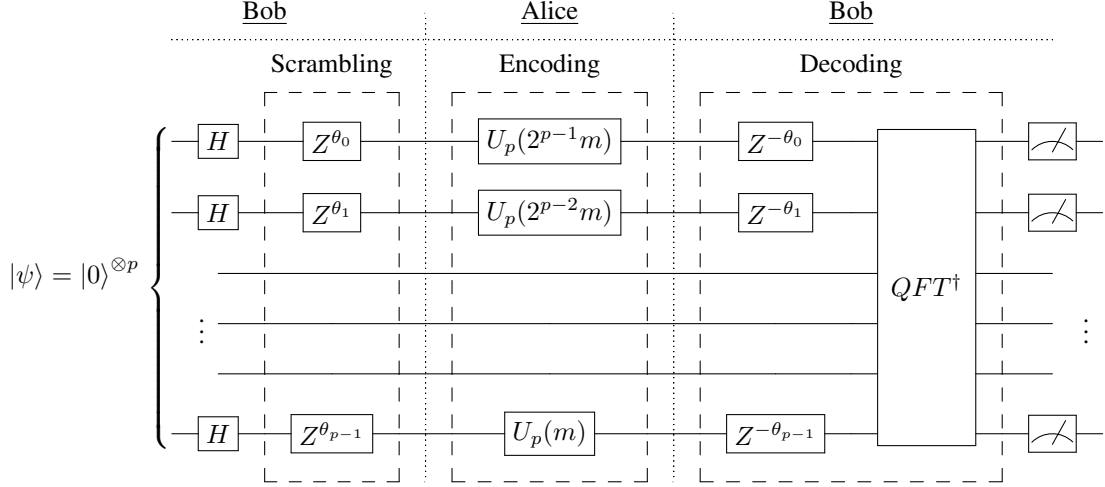
\begin{figure}[h]
    \centering
    \hspace*{.7cm}
    \centerline{
    \Qcircuit @C=.7em @R=1em {
    && \mbox{\underline{Bob}} &&&  && \mbox{\underline{Alice}} &&  &&& \mbox{\underline{Bob}} &&& \\
    &&&&&&&&&&&&&&& \\
    &&& \mbox{Scrambling} &&&& \mbox{Encoding} &&&&& \mbox{Decoding} &&& \\
    &&&&&&&&&&&&&&& \\
    \lstick{} & \gate{H} & \qw & \gate{Z^{\theta_0}} & \qw & \qw & \qw & \gate{U_p(2^{p-1}m)} & \qw & \qw & \qw & \gate{Z^{-\theta_0}} & \qw & \multigate{5}{QFT^\dagger} & \qw & \meter & \qw \\
    \lstick{} & \gate{H} & \qw & \gate{Z^{\theta_1}} & \qw & \qw & \qw & \gate{U_p(2^{p-2} m)} & \qw & \qw & \qw & \gate{Z^{-\theta_1}} & \qw & \ghost{QFT^\dagger} & \qw & \meter & \qw \\ 
    & \lstick{} & \qw & \qw & \qw & \qw & \qw & \qw & \qw & \qw & \qw & \qw & \qw & \ghost{QFT^\dagger} & \qw & \qw & \lstick{} \\ 
    & \lstick{\vdots} & \qw & \qw & \qw & \qw & \qw & \qw & \qw & \qw & \qw & \qw & \qw & \ghost{QFT^\dagger} & \qw & \qw & \lstick{\vdots} \\ 
    & \lstick{} & \qw & \qw & \qw & \qw & \qw & \qw & \qw & \qw & \qw & \qw & \qw & \ghost{QFT^\dagger} & \qw & \qw & \lstick{} \\ 
    \lstick{} & \gate{H} & \qw & \gate{Z^{\theta_{p-1}}} & \qw & \qw & \qw & \gate{U_p(m)} & \qw & \qw & \qw & \gate{Z^{-\theta_{p-1}}} & \qw & \ghost{QFT^\dagger} & \qw & \meter & \qw \\
    &&&&&&&&&&&&&& \\
    \protected\gategroup{2}{1}{2}{16}{0em}{.}
    \protected\gategroup{1}{6}{11}{6}{0em}{.}
    \protected\gategroup{4}{3}{11}{5}{0em}{--}
    \protected\gategroup{4}{7}{11}{9}{0em}{--}
    \protected\gategroup{1}{10}{11}{10}{0em}{.}
    \protected\gategroup{4}{11}{11}{15}{0em}{--}
    \inputgroupv{5}{10}{1.0em}{5.0em}{\ket{\psi} = \ket{0}^{\otimes p}\quad\quad\quad\quad} 
    }}
    \caption{The circuit for the quantum portion of the QFT-based QKD protocol presented in this paper.}
    \label{fig:twoPassCircuit}
\end{figure}

\paragraph*{\textbf{Scrambling}}
First, Alice and Bob determine the verification bit string $v$ along with the positions the bits will occupy in the transmitted messages.
This scheme, unlike BB84, does not require any secrecy regarding the verification bits, so we assume this is public information to which Eve has access.
Therefore, this scheme allows for an entirely quantum key distribution protocol since no classical communication is needed for verification.
Alternatively, Alice may randomly create a verification bit string in secret, similar to in BB84, but this requires classical communication between the two parties afterwards to verify it.
Possible schemes for verification are discussed in Sec.~\ref{sec:schemes}.

Bob creates a state with the quantum register $\ket{\psi} = \ket{+}^{\otimes p}$, where $p = |k| + |v|$, $|k|$ is the number of bits in the key to be shared, and $|v|$ is the number of verification bits. 
Here, $\ket{+} = H\ket{0} = \frac{1}{\sqrt{2}} (\ket{0} + \ket{1})$, where $H = \frac{1}{\sqrt{2}}\begin{bmatrix}
    1 & 1 \\ 1 & -1
\end{bmatrix}$.
Next, he applies $p$ random phase shifts, 
\begin{equation}
    \text{Scr}(\theta) = \bigotimes_{j=0}^{p - 1} Z^{\theta_{p - 1 - j}},
\end{equation}
where $\theta_j$ is the random phase applied to the $j$-th qubit and $Z^{\theta_j} = \begin{bmatrix}1 & 0 \\ 0 & e^{i \theta_j}\end{bmatrix}$. This results in the register $\ket{\psi} = \frac{1}{\sqrt{2^{p}}}(\ket{0} + e^{i \theta_{p - 1}} \ket{1})(\ket{0} + e^{i \theta_{p - 2}} \ket{1})...(\ket{0} + e^{i \theta_{0}} \ket{1})$.
Bob sends this state to Alice.

Preparing these bits with random phases that only Bob knows ensures that, upon interference by Eve, Eve will destroy and cannot replicate these phases.
In that case, when Bob reverses the random phases in the decoding step, the state needed to keep the $QFT^\dagger$ from causing entanglement will be destroyed, thereby causing entanglement and exposing the interference.
In some sense, these random phases may be considered a secret key, known only to Bob, that Alice can use without knowledge of it.

\paragraph*{\textbf{Encoding}}
Alice generates a bit string $k$, which is the key she will share with Bob. 
Alice then inserts the verification bits into either predetermined or randomly generated positions, resulting in the bit string $m$, which may also be interpreted as an integer value.
She then encodes $m$ by applying 
\begin{equation} \label{eq:encoding}
    \text{Enc}(m) = \bigotimes_{j=0}^{p - 1} U_p(2^{j} m),
\end{equation}
where $U_p(x) = Z^{\pi x / 2^{p -  1}}$.
Alice sends this state to Bob.

\paragraph*{\textbf{Decoding}}
Bob undoes his rotations, followed by $QFT^\dagger$, to retrieve $m'$
\begin{equation}
    \text{Dec}(\theta) = QFT^\dagger \bigotimes_{j=0}^{p - 1} Z^{-\theta_{p-1-j}}.
\end{equation}
The circuit for the steps up until this point is shown in Fig.~\ref{fig:twoPassCircuit}.

\paragraph*{\textbf{Extraction}}
Bob can then extract the verification bit positions, either predetermined or classically communicated after measurement, from $m'$ resulting in $v'$.
The remaining bits in $m'$ then form $k'$.
If $v' = v$, or if $v'$ differs from $v$ by less than some agreed-upon limit, Bob may assume that Eve did not interfere with the state and that $k' = k$, meaning that Alice and Bob now share a key $k$.
Otherwise, they throw the key out and try again.

Note that this circuit resembles the Quantum Phase Estimation circuit~\cite{nielsen_chuang_2010} once Bob's secret phases are undone, ensuring that Bob can reliably recover the key $k$ when there is no eavesdropper.

\subsection{Interference Resistance}
In order to always evade detection, Eve must measure in the basis dependent on Bob's (secret) random phases.
It is extremely unlikely that Eve guesses the basis of the random phases that Bob initially encoded.
In fact, the probability that Eve guesses correctly is 0 if Bob chooses rotations from a continuous distribution.
Therefore, any measurement Eve makes will collapse to a different state than Bob's initial state with probability 1.
When he undoes his state in the decoding stage, the measured qubits will have random phases, causing the $QFT^\dagger$ to apply random rotations to subsequent qubits, increasing the chance of mismatched verification qubits. 
This is explored in~\cite{nagy2008key} and Sec.~\ref{sec:probability}. 
Thus, it is difficult for Eve to avoid detection by Bob in any intercept-resend attack.
An optional authenticated classical channel can be used once Bob obtains the key to verify that the key was from Alice instead of Eve, thereby preventing an impersonation attack as well.

\subsection{Direct Message Encryption} \label{sec:qft-enc}

\begin{figure}
    \centering
    \hspace*{.3cm}
    \centerline{
    \Qcircuit @C=.45em @R=1em {
    &&&& \mbox{\underline{Alice}} &&& \mbox{\underline{Bob}} && \mbox{\underline{Alice}} &&& \mbox{\underline{Bob}} &&&& \\
    &&&&&&&&&&&&&&&& \\
    \lstick{} & \gate{H} & \qw & \gate{U_p(2^{p-1}m)} & \qw & \gate{Z^{\theta_0}} & \qw & \gate{Z^{\phi_0}} & \qw & \gate{Z^{-\theta_0}} & \qw & \gate{Z^{-\phi_0}} & \qw & \multigate{5}{QFT^\dagger} & \qw & \meter & \qw \\
    \lstick{} & \gate{H} & \qw & \gate{U_p(2^{p-2} m)} & \qw & \gate{Z^{\theta_1}} & \qw & \gate{Z^{\phi_1}} & \qw & \gate{Z^{-\theta_1}} & \qw & \gate{Z^{-\phi_1}} & \qw & \ghost{QFT^\dagger} & \qw & \meter & \qw \\
    & \lstick{} & \qw & \qw & \qw & \qw & \qw & \qw & \qw & \qw & \qw & \qw & \qw & \ghost{QFT^\dagger} & \qw & \qw & \lstick{} \\ 
    & \lstick{\vdots} & \qw & \qw & \qw & \qw & \qw & \qw & \qw & \qw & \qw & \qw & \qw & \ghost{QFT^\dagger} & \qw & \qw & \lstick{\vdots} \\ 
    & \lstick{} & \qw & \qw & \qw & \qw & \qw & \qw & \qw & \qw & \qw & \qw & \qw & \ghost{QFT^\dagger} & \qw & \qw & \lstick{} \\ 
    \lstick{} & \gate{H} & \qw & \gate{U_p(m)} & \qw & \gate{Z^{\theta_{p-1}}} & \qw & \gate{Z^{\phi_{p-1}}} & \qw & \gate{Z^{-\theta_{p-1}}} & \qw & \gate{Z^{-\phi_{p-1}}} & \qw & \ghost{QFT^\dagger} & \qw & \meter & \qw \\
    &&&&&&&&&&&&&&&& \\
    \protected\gategroup{2}{1}{2}{17}{0em}{.}
    \protected\gategroup{1}{7}{9}{7}{0em}{.}
    \protected\gategroup{1}{9}{9}{9}{0em}{.}
    \protected\gategroup{1}{11}{9}{11}{0em}{.}
    \inputgroupv{3}{8}{1em}{5em}{\ket{0}^{\otimes p}\quad} 
    }
    }
    \caption{The circuit for direct message encryption using the QFT.}
    \label{fig:threePassCircuit}
\end{figure}
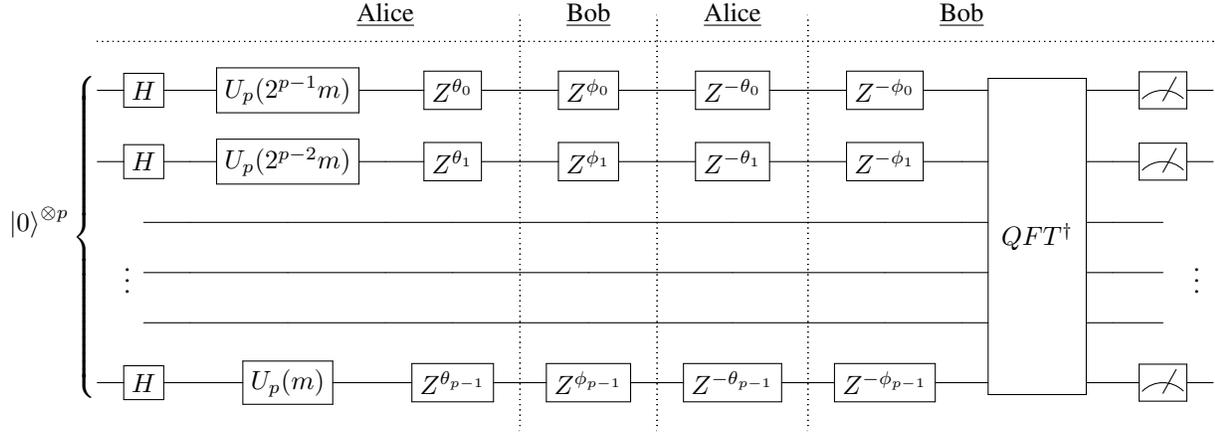

The QKD method introduced in Sec.~\ref{sec:qft-qkd} should not, however, be used to directly send sensitive information since Eve can measure Bob's initial transmission and know the state of the system when Alice receives it. 
Eve can then measure Alice's transmission in the QFT basis minus the measured state to retrieve Alice's string.
Since the verification bit positions may be public information, Eve could then know $k$ herself.
Instead, we introduce a QFT-based three-pass protocol here that can be used to directly send an encrypted message while preventing partial leakage.

Specifically, the protocol proceeds according to Kak's protocol~\cite{threeStage}, except the secret state to transfer is encoded by inserting the verification bits in their respective positions.
This results in the string $m$ to which the $\text{Enc}(m)$ transformation (Eq.~\ref{eq:encoding}) is applied and retrieved at the end with $QFT^{\dagger}$. 
Since Eve cannot measure before the message was concealed with Alice's secret phases and released, Eve is unlikely to recover any meaningful phase information unless she knows the secret phases.

The analyses of verification bit schemes discussed in Secs.~\ref{sec:probability} and~\ref{sec:schemes} apply to this three-pass scheme when $U_A = Scr(\theta)$ and $U_B = Scr(\phi)$, where $\theta_i$ is Alice's random phase on the $i$-th qubit and $\phi_i$ is Bob's.
The three-pass circuit with these unitary operations is shown in Fig.~\ref{fig:threePassCircuit}.

\section{Probability of Interference Detection}\label{sec:probability}

We may analyze the effectiveness of different verification layouts and schemes by calculating the likelihood of detecting interference from Eve.
A verification qubit contributes to this detection by the likelihood that it measures incorrectly, i.e., it measures as zero when intended to be one or vice versa.
This probability is determined by the error in phase rotation before the final Hadamard gate, $H$, on that qubit in the $QFT^\dagger$ gate.

In the below analysis we assume no transmission noise as well as no errors in gates, qubits, or circuit implementation.
In practice, one issue with the QFT encoding is the propagation of errors.
Especially during the Noisy Intermediate Scale Quantum (NISQ) computing era, errors may propagate and be misinterpreted as Eve interference with a higher probability than BB84.
The downside of this proposed scheme, then, is more often having to restart the protocol unnecessarily.
We leave such a noise-sensitive protocol analysis for future work.

\subsection{Probability of measuring correctly due to a phase error}
To always measure a qubit as zero in the Hadamard basis, the relative phase rotation of that qubit needs to be exactly zero.
So, the state before the final Hadamard in the $QFT^\dagger$ gate should be $\frac{1}{\sqrt{2}} (\ket{0} + \ket{1})$.
With an error in the phase rotation, say $\theta_e$, we have the state after the Hadamard gate, $H$, application as
\begin{equation}
    \frac{1}{\sqrt{2}} H (\ket{0} + e^{i \theta_e} \ket{1}) = \frac{1}{2}((1 + e^{i \theta_e})\ket{0} + (1 - e^{i \theta_e})\ket{1}).
\end{equation}
Hence, the probability that the qubit measures correctly as zero given some rotation error is 
\begin{equation}
    \Pr_{r,0}(\theta_e) = \left| \frac{1}{2}(1 + e^{i \theta_e}) \right|^2.
\end{equation}

In the case where the intended measurement is one, the state before the final Hadamard should be $\frac{1}{\sqrt{2}} (\ket{0} + e^{i \pi} \ket{1})$.
With an error in the phase rotation, say $\theta_e$, we have the state after the Hadamard gate application as $\frac{1}{2}((1 + e^{i (\pi + \theta_e)})\ket{0} + (1 - e^{i (\pi + \theta_e)})\ket{1})$.
Hence, the probability that the qubit measures correctly as one given some rotation error is
\begin{equation}
    \begin{aligned}
        \Pr_{r,1}(\theta_e) &= \left| \frac{1}{2}(1 - e^{i (\pi + \theta_e)}) \right|^2 \\
        &= \left| \frac{1}{2}(1 - e^{i \pi} e^{i \theta_e}) \right|^2 \\
        &= \left| \frac{1}{2}(1 + e^{i \theta_e}) \right|^2 \\
        &= \Pr_{r,0}(\theta_e).
    \end{aligned}
\end{equation}
Therefore, the probability that the qubit measures correctly is the same, given some phase error, whether the intended measurement is supposed to be zero or one.
So, we use the simpler expression as the equation to determine the probability of measuring correctly given some rotation error, $\Pr_r(\theta_e)$, as
\begin{equation}
    \begin{aligned}
        \Pr_r(\theta_e) &= \Pr_{r,0}(\theta_e) \\
        &= \left| \frac{1}{2}(1 + e^{i \theta_e}) \right|^2.
    \end{aligned}
\end{equation}

\subsection{Rotation effect from previous qubits}
Since $QFT^\dagger$ can be considered as each previous qubit having been measured before applying the controlled phase rotation~\cite{hiddenSubgroup}, the phase error that will be applied to a particular qubit, which we refer to as the target qubit, depends on the previous qubits that measured incorrectly.
If there are no errors in previous qubits, then the target qubit will have no additional phase error and hence measure correctly if not touched by the attacker.
If there is an error in a previous qubit, then the rotation effect depends on how far away the incorrect qubit is from the target qubit as well as whether the incorrect qubit was intended to be one or zero.

If the intention was zero, then it will have incorrectly measured one, triggering the controlled phase rotation in $QFT^\dagger$ and adding a negative phase error to the target.
If the intention was one, then it was supposed to trigger the controlled phase rotation but instead did not, effectively adding a positive phase error to the target.

The magnitude of the rotation is directly due to the phase rotation applied in the $QFT^\dagger$ gate.
Denote the actual measured outcome $A$ of $n$ previous qubits as
a binary array where $a_i = 1$ means that the qubit $i$ was measured as one, and $a_i = 0$ means that it was measured as zero.
We also denote the intended outcome of the qubit measurements as a binary array $B$, where $b_i \in \{ 0, 1 \}$ denotes the intended outcome of qubit $q_j$.
We say that the rotation effect applied to a qubit $q_t$ due to a previous qubit $q_j$, where $j < t$, is
\begin{equation}
    \label{eq:rotation}
    \begin{aligned}
        R(t, j, A, B) &= (b_j - a_j) \frac{\pi}{2^{t-j}}.
    \end{aligned}
\end{equation}

Finding the probability of a target qubit measuring correctly due to the rotation effect from previous qubits depends, then, on the intended and actual measurement outcomes of the previous qubits.
More precisely, the cumulative rotation error $\theta_e$ on a qubit $q_t$ given intended measurements $B$ and actual measurements $A$ is computed as
\begin{equation}
    \theta_e(t, A, B) = \sum_{j=0}^{t - 1} R(t, j, B, A).
\end{equation}

\subsection{Probablity of measuring a qubit correctly}

The probability that the previous qubits to the target qubit $q_t$ lie in the outcome $A$ is
\begin{equation}
    \begin{aligned}
        \Pr_p(t, A, B) &= \prod_{j=0}^{t - 1} \begin{cases}
            \Pr_c(j, B) & a_j = b_j \\
            1 - \Pr_c(j, B) & \text{otherwise},
        \end{cases}
    \end{aligned}
\end{equation}
 where $\Pr_c(j, B)$ is the probability of qubit at index $j$ measuring correctly.
This probability is determined by summing, for each possible measurement outcome of previous qubits, the product of the probability of obtaining that measurement outcome and the probability of measuring the target qubit correctly given the resulting rotation effect,
\begin{equation}
    \label{eq:correctnessProbability}
    \begin{aligned}
        \Pr_c(t, B) &= \sum_{A \in \{0,1\}^{t}} \Pr_p(t, A, B) \Pr_r(\theta_e(t, A, B)).
    \end{aligned}
\end{equation}

There are two special cases to this equation. 
(1) If the target qubit has been measured by Eve, then we simply say that $\Pr_c(t, B) = .5$.
This is because since Eve's best option is to ``undo'' some random rotation and measure in the X-basis, by expectation there will be an equal chance of measuring correctly or incorrectly after the subsequent $QFT^\dagger$ gate regardless of phase errors induced by previous measurements.
This is also explored in~\cite{nagy2008key}.
(2) If the target bit is zero and Eve did not measure it, then $\Pr_c(t, B) = 1$.

One can then determine the total probability of detecting Eve, here defined as the probability that at least one verification bit is measured incorrectly. 
This probability, given some $B$ and the bit indices of the verification bits, $V$, is

\begin{equation}
    \label{eq:detection-chance}
    \Pr_e(V, B) = 1 - \prod_{t \in V} \Pr_c(t, B).
\end{equation}

Note that a consequence of our protocol is that phase error contributions from different qubits can have different signs, and can therefore diminish the collective phase error.
For example, suppose $q_1$ incorrectly measured as 1, leading to a negative phase error contribution on future qubits, and suppose $q_3$ incorrectly measured as 0, leading to a positive phase error contribution on future qubits.
Then the collective phase error on $q_4$ is $-\frac{\pi}{2^3} + \frac{\pi}{2^1}$, a diminished phase error due to the opposing signs of $q_1$ and $q_3$.
As a consequence, the probability of detecting an attacker depends on the intended measurement outcome $B$.

Since the probability of detection is dependent on the specific $B$, it can then be computed by ranging over a set of possible intended measurement outcomes.
One way this overall chance may be computed is by averaging the obtained probabilities. 
However, it is safer to consider the case that phase errors introduced by the eavesdropper interfere with each other leading to the worst likelihood of detection.
This worst-case analysis can be accomplished by instead taking the minimum probability of detection, where the qubits are most likely to measure correctly despite any phase errors.

Note that this formula takes exponential time to compute.
We opt to analytically compute probabilities of detection for small numbers of qubits in Fig.~\ref{fig:graphs}.
However, when comparing verification schemes for implementations it may be preferable to instead simulate the probability of detection using the parameters of the physical implementation and scheme.

\section{Analysis of Verification Schemes}\label{sec:schemes}
\begin{figure*}
    \centering
    \includegraphics[width=\linewidth]{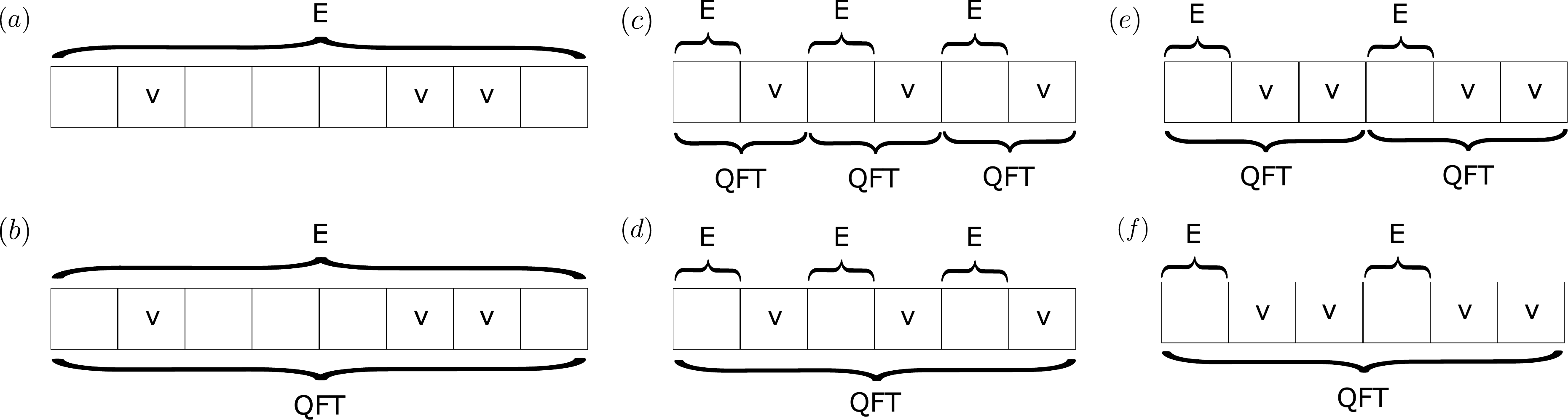}
    \caption{Example verification schemes. Each square represents a transmitted qubit. Qubits marked with $v$ are used for verification, and qubits marked with $E$ are qubits that Eve measured. Qubits marked with $QFT$ denote compartments of QFT encoding. (a) BB84 with random verification qubit placement and Eve measuring all qubits. (b) Two-pass, QFT-based QKD with random verification qubit placement and Eve measuring all qubits. (c) Compartmentalized two-pass, QFT-based QKD with one verification qubit per key qubit. (d) Uncompartmentalized two-pass, QFT-based QKD with one verification qubit per key qubit. (e) Compartmentalized two-pass, QFT-based QKD with two verification qubits per key qubit. (f) Uncompartmentalized two-pass, QFT-based QKD with two verification qubits per key qubit. Eve is assumed to only measure key qubits in (c)-(f).}
    \label{fig:verificationSchemes}
\end{figure*}

The probability of detecting interference from an eavesdropper is dependent on the verification scheme used.
We first discuss an analogous verification scheme to the one used in BB84~\cite{bb84}.
Notably our protocol with this scheme provides improvements against both complete and partial attacks.
We then present a collection of public verification schemes and analyze their theoretical detection performance.
Notably the probability of detecting an eavesdropper using public verification schemes in protocols like BB84 is 0, while our protocol still allows for detecting the eavesdropper in these cases.

\subsection{Analogous Verification Scheme to BB84}
The verification scheme in BB84 involves random post-selection of verification qubits.
Specifically, Alice and Bob agree on a random subset of qubits to verify.
One such choice of random qubits is shown in Fig.~\ref{fig:verificationSchemes}(a).
We can therefore describe an analogous verification scheme for the protocol introduced in Sec.~\ref{sec:qft-qkd} where the same verification qubits are chosen as shown in Fig.~\ref{fig:verificationSchemes}(b).
To present a fair comparison, we assume the attacker does not know which bits are verification bits \textit{a priori}.
This can be achieved by Alice secretly choosing the verification qubit positions when constructing her message and then later conveying those positions to Bob after his measurements.

We may begin analyzing the chance of detection in the simple case that Eve measures all transmitted qubits, the typical attacker model when analyzing QKD schemes.
BB84, shown in Fig.~\ref{fig:verificationSchemes}(a), has a $1 - (.75)^{|v|}$ chance of detecting Eve when all the verification qubits are measured.
This is because Eve is only detectable if she both picks the incorrect basis for a qubit (probability $0.5$) and Bob measures the wrong outcome for that qubit (probability $0.5$), resulting in only a probability of detection of $0.25$ for each individual verification qubit.
However, in our proposed protocol with an analogous verification scheme, shown in Fig.~\ref{fig:verificationSchemes}(b), each verification qubit will have a $.5$ chance of measuring correctly (Eq.~ \ref{eq:correctnessProbability}).
This results in an overall probability of $1 - (.5)^{|v|}$ of Eve being detected.
In this case, QDK using QFT outperforms BB84 significantly.

It is also the case that our proposed QKD protocol outperforms BB84 in detecting an eavesdropper when Eve measures only a subset of the transmitted qubits.
This scenario corresponds to the partial information leakage attack mentioned in Sec.~\ref{sec:partial-information}.
Let the chance of detecting an eavesdropper be defined by Eq.~\ref{eq:detection-chance}, except where $\Pr_c(t, B)$, as defined in the QFT case by Eq.~\ref{eq:correctnessProbability}, can be differently defined for different protocols. 
For clarity we use $\Pr_{e,BB84}(V, B)$ to denote probabilities for the BB84 protocol, and $\Pr_{e,QFT}(V, B)$ to denote probabilities for our protocol.

\begin{lemma}
\label{lemma:detection-as-good-as-bb84}
The probability of detection for our protocol is at least as great as that for BB84: $\Pr_{e,BB84}(V, B) \leq \Pr_{e,QFT}(V, B)$
\end{lemma}

\begin{proof}[Proof sketch]
    For each qubit, assume Eve measures a qubit from our protocol if and only if Eve measures the corresponding qubit in the BB84 protocol (proof by coupling).

    By the definition of $\Pr_e(V, B)$,
    ${\prod_{t \in V} \Pr_{c,BB84}(t, B)
    \geq \prod_{t \in V} \Pr_{c,QFT}(t, B)}
    \implies
    {\Pr_{e,BB84}(V, B) \leq \Pr_{e,QFT}}(V, B)$.
    
    We then consider each qubit by ascending index $i$.
    Since we assume noiseless protocols, let $\Pr_c(i, B) = 1$ unless otherwise specified.
    \begin{case}[Eve does not measure $q_i$, $q_i$ is not a verification qubit]
        No effect in this case.
    \end{case}
    \begin{case}[Eve measures $q_i$, $q_i$ is not a verification qubit] \label{case:eve-not-verif}
        In BB84, non-verification qubits are independent from verification qubits. 
        As a result, measuring non-verification qubits does not effect verification qubits.

        In our protocol, this could potentially cause a phase rotation on subsequent verification qubits, so $\forall j. \; j > i \implies \Pr_{c,QFT}(j, B) \leq 1$.
        In other words, measuring a non-verification qubit may introduce a probability of error in later verification qubits.
    \end{case}
    \begin{case}[Eve does not measure $q_i$, $q_i$ is a verification qubit] \label{case:eve-no-measure-verif}
        For BB84, verification qubits are independent from non-verification qubits. As a result, there is no effect on verification qubits if it is untouched by Eve.

        For our protocol, if no previous qubit was measured by Eve, then there is no potential measurable effect. In this case $\Pr_{c,QFT}(i, B) = 1$, equal to the probability of evasion in BB84. However, if a previous qubit was measured by Eve, then by Case \ref{case:eve-not-verif} $\Pr_{c,QFT}(i, B) \leq 1$, meaning the probability of evasion for this qubit under our protocol is less than the probability of evasion for this qubit under BB84.

        Overall for this case, it is clear that $\Pr_{c,QFT}(i, B) \leq \Pr_{c,BB84}(i, B)$
    \end{case}
    \begin{case}[Eve measures $q_i$, $q_i$ is a verification qubit] \label{case:eve-measure-verif}
        For BB84, $\Pr_{c,BB84}(i, B) = .75$ as previously discussed.
        For our protocol, $\Pr_{c,QFT}(i, B) = .5$ as previously discussed.
        It is therefore clear that $\Pr_{c,QFT}(i, B) < \Pr_{c,BB84}(i, B)$ for this case.
    \end{case}

    Since ${\forall t \in V. \; {0 \leq \Pr_c(t, B) \leq 1}}$ and by Cases \ref{case:eve-no-measure-verif} and \ref{case:eve-measure-verif} 
    \begin{equation}
        \forall t \in V. \; {\Pr_{c,QFT}(t, B) \leq \Pr_{c,BB84}(t, B)},
    \end{equation} then $\prod_{t \in V} \Pr_{c,BB84}(t, B)
    \geq \prod_{t \in V} \Pr_{c,QFT}(t, B)$.
    Finally, by the definition of $\Pr_e$,
    \begin{equation}
        \Pr_{e,BB84}(V, B) \leq \Pr_{e,QFT}(V, B).
    \end{equation}
\end{proof}

\begin{figure*}[h]
    \centering
    \includegraphics[width=\linewidth]{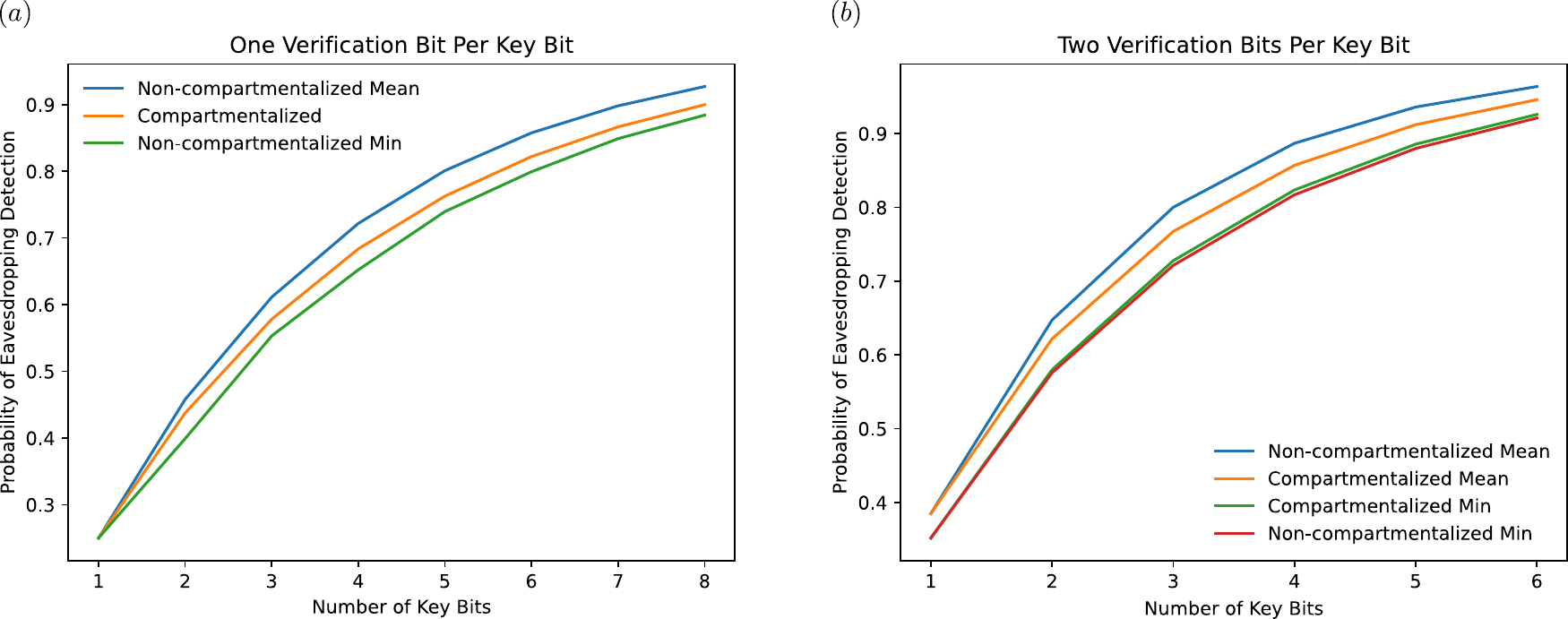}
    \caption{The probabilities of detecting an eavesdropper over an increasing number of key qubits with (a) one and (b) two verification qubits per key qubit.}
    \label{fig:graphs}
\end{figure*}

\subsection{Public Verification Schemes}

Unlike BB84 and other QKD protocols with independent key and verification qubits, our protocol is resilient against an attacker even if they know the locations and intended values of the verification qubits.
In this section we will therefore examine public verification schemes.
Specifically we analyze representative schemes in which all verification qubits are intended to measure zero.

We analytically compute the probabilities of detection for four different public verification schemes in which key qubits are uniformly interleaved with verification qubits.
In two schemes, shown in Figs.~\ref{fig:verificationSchemes}(c) and (d), each key qubit is followed by one verification qubit.
An eavesdropper's optimal strategy then is to measure only the key qubits and leave the verification qubits untouched as indicated by the \texttt{E} over each key qubit in the figure.
This knowledge would break most QKD schemes, including BB84. 
However, our protocol, based on the QFT transform, can still detect an eavesdropper in this scenario.
Alternative schemes include one where each key qubit is followed by two verification qubits (Figs.~\ref{fig:verificationSchemes}(e) and (f)).

These schemes can be further modified by what we call compartmentalization.
That is, instead of encoding and decoding all the bits with a single QFT, the message would be split into smaller groups that are independently encoded and decoded with the QFT.
These groups are denoted by several \texttt{QFT}s in Fig.~\ref{fig:verificationSchemes}.

For example, Fig.~\ref{fig:verificationSchemes}(c) shows a scheme where each group consists of a key and verification qubit pair.
Notably, this public verification scheme matches the chance of detection from BB84 when all qubits are measured, $1 - (.75)^{|v|}$ as per Eq.~\ref{eq:correctnessProbability}.
This is quite interesting because the worst case for Alice and Bob in this public scheme, where Eve deliberately chooses her measured qubits with knowledge, is equal to the best case for Alice and Bob in BB84, where Eve measures all qubits.
This emphasises the benefit of using a QFT-based protocol to increase eavesdropping detection.

As mentioned in Sec.~\ref{sec:probability}, we can compute the average  or pessimistic probability of detection using the \texttt{mean} or \texttt{min} statistic functions, respectively.
When comparing schemes that use compartmentalization versus ones that do not, one can see that the statistics function used makes a difference.
In Fig.~\ref{fig:graphs}(a), for example, we compare the compartmentalized scheme from Fig.~\ref{fig:verificationSchemes}(c) against the non-compartmentalized scheme from Fig.~\ref{fig:verificationSchemes}(d).
One can see that the non-compartmentalized version is better in the average case, using the $\texttt{mean}$ function, but worse in the worst-case, using the $\texttt{min}$ function.
This is due to the fact that, in the worst-case, the positive and negative phase errors across all the qubits interfere with each other and effectively decrease the final phase error in the verification qubits.
On the other hand, in the average case, the phase errors across the individual qubits could amplify the final phase error on a subsequent verification qubit.
In contrast, the compartmentalized version limits the number of entangled qubits that could contribute to a phase error on the verification qubit, leading to a lower chance of detection compared to the non-compartmentalized average-case.
Notably, since it is infeasible to exactly compute these probabilities with large numbers of qubits, it is not clear how significant this advantage will be as the number of qubits grows.

Looking at Fig.~\ref{fig:graphs}(b), which compares the compartmentalized scheme from Fig.~\ref{fig:verificationSchemes}(e) against the non-compartmentalized scheme from Fig.~\ref{fig:verificationSchemes}(f), we again see that the non-compartmentalized scheme outperforms the compartmentalized one on average.
However, in the worst-case the compartmentalized and non-compartmentalized schemes have similar detection rates since the worst-case interference pushes the phase errors towards zero in both cases, which we call dampening.
Though similar, one can see that the compartmentalized detection rate is slightly higher as it limits possible dampening by reducing the number of subsequent qubits that can negate the effect of earlier phase errors.

One can see that different verification schemes can be better for different use cases, and
Eq.~\eqref{eq:correctnessProbability} allows one to analyze different schemes to try to pick one that is best-suited for their specific case.
If two parties were sending many messages, where each message in particular didn't need high security, but they wanted to detect Eve with high likelihood over the course of many messages, then they might consider using QFT encoding on all the qubits together in a non-compartmentalized way.
If they wanted high security on each message, on the other hand, they may want to use a compartmentalized version of the scheme.

In all these schemes, the chance of detecting an attacker is below $1 - (.5)^{|v|}$, which is the chance of detecting an attacker who measures all qubits, as is generally assumed in the case without knowledge of the verification scheme.
However, unlike BB84 where the chance of detecting an attacker given a public verification scheme is 0, with our protocol the chance of detecting an attacker grows with each qubit they measure.

\section{Conclusion}
We have shown a novel two-pass approach to QKD using the QFT as well as a QFT-based approach to sending encrypted messages directly.
These protocols allow for a greater chance of detecting an eavesdropper against certain attacks, including partial information leakage and many-copies attacks, than independent qubit measurement based approaches such as BB84.
This paper then provides a formula to determine the probability of detecting an eavesdropper's presence under any verification scheme.
As we showed through key examples, this formula allows for analyzing verification schemes of one's own to help determine the best scheme for each situation.
Notably, our analysis shows that our QFT-based approach allows for detecting an eavesdropper even when the verification scheme is known to the attacker.
Altogether our protocols provide for robust eavesdropping detection against more powerful eavesdroppers than other schemes can handle.
We leave it to future work to validate that implementations of our protocols on real quantum devices match our theoretical results.

\section*{Acknowledgments}
The authors thank John Drew Wilson for useful discussions.



\end{document}